\documentclass[11pt]{article}

\usepackage{natbib}
\usepackage{hyperref}
\usepackage{amssymb,amsmath,xcolor,graphicx,xspace,float,colortbl,rotating} %
\usepackage{textcomp}
\usepackage{appendix}  
\usepackage{bm}  
\usepackage{boxedminipage}  
\usepackage{color}  
\usepackage{endnotes}  
\usepackage[onehalfspacing]{setspace}  
\usepackage{tabulary}  
\usepackage{varioref}  
\usepackage{wrapfig}  
\usepackage[linesnumbered,ruled]{algorithm2e}
\usepackage{tikz}
\usepackage{subcaption}
\usepackage{float}
\restylefloat{table}

\usepackage[margin=1in]{geometry}
\newtheorem {theorem}{Theorem}

\newtheorem {assumption}{Assumption}

\newtheorem {definition}{Definition}

\newtheorem {proposition}{Proposition}
\newtheorem {remark}{Remark}

\newenvironment {proof}[1][Proof]{\noindent \textbf {#1.} }{\ \rule {0.5em}{0.5em}}

\title{Bounded Foresight Equilibrium in Large Dynamic Economies with Heterogeneous Agents and Aggregate Shocks}
\author{Bilal Islah\footnote{Africa Business School, University Mohammed VI Polytechnic, Morocco. e-mail: \href{mailto:bilal.islah@um6p.ma}{bilal.islah@um6p.ma}} \hspace{5em} Bar Light\footnote{Business School and Institute of Operations Research and Analytics, National University of Singapore, Singapore. e-mail: \href{mailto:barlight@nus.edu.sg}{barlight@nus.edu.sg}}}
\date{February 2025}
\begin{document}
\maketitle

\begin{abstract}

Large dynamic economies with heterogeneous agents and aggregate shocks are central to many important applications, yet their equilibrium analysis remains computationally challenging. This is because the standard solution approach, rational expectations equilibria require agents to predict the evolution of the full cross-sectional distribution of state variables, leading to an extreme curse of dimensionality. 
In this paper, we introduce a novel equilibrium concept, N-Bounded Foresight Equilibrium (N-BFE), and establish its existence under mild conditions. In N-BFE, agents optimize over an infinite horizon but form expectations about key economic variables only for the next 
$N$ periods. Beyond this horizon,  they assume that economic variables remain constant and use a predetermined continuation value. This equilibrium notion reduces computational complexity and draws a direct parallel to lookahead policies in reinforcement learning, where agents make near-term calculations while relying on approximate valuations beyond a computationally feasible horizon. At the same time, it lowers cognitive demands on agents while better aligning with the behavioral literature by incorporating time inconsistency and limited attention, all while preserving desired forward-looking behavior and ensuring that agents still respond to policy changes.  Building on this framework, we conduct a numerical analysis of N-BFE in classical models from the heterogeneous agent macro literature and large dynamic competition models. Importantly, in N-BFE equilibria, forecast errors arise endogenously. We measure the foresight errors for different foresight horizons and show that foresight significantly influences the variation in endogenous equilibrium variables, distinguishing our findings from traditional risk aversion or precautionary savings channels. This variation arises from a feedback mechanism between individual decision-making and equilibrium variables, where increased foresight induces greater non-stationarity in agents’ decisions and, consequently, in economic variables.

\end{abstract}

\doublespacing

\break

\section{Introduction}

In many economic settings, heterogeneous agents or firms make decisions over time, incurring payoffs that depend on both their choices and evolving economic variables such as market prices or total production. These environments are typically modeled as large dynamic economies with idiosyncratic and aggregate states, where individual decisions shape key economic variables, which, in turn, influence the decisions of all participants. This feedback loop creates intricate interdependencies, making the analysis of such systems both conceptually rich and computationally challenging. A key difficulty in these models is that equilibrium prices and other key economic  variables are non-stationary and are evolving endogenously with the economy rather than converging to a fixed steady state, as in stationary models without aggregate states. 

Traditionally, the rational expectations equilibrium has been the standard approach for solving such models. In this approach, agents make decisions based on all available information and correctly forecast the distribution of future economic variables. Rational expectations have been widely used to analyze policy interventions and equilibrium dynamics. However, in large-scale heterogeneous-agent models with aggregate states, that have become popular in recent literature, this approach is generally not computationally feasible even for simple models. This is because rational expectations require agents to forecast future equilibrium prices by predicting the full cross-sectional distribution of state variables, leading to an extreme version of the curse of dimensionality. This results in a dynamic optimization problem for the agents or firms in which the entire cross-sectional distribution becomes a state variable, making the computation of equilibrium highly complex and, in many cases, infeasible. Recently, \cite{Moll2024} provided an excellent analysis of the limitations of rational expectations in such settings, arguing that the assumption that agents forecast equilibrium prices by tracking correctly the full  state distribution imposes computational intractability and unrealistic cognitive demands. These challenges severely limit the applicability of  heterogeneous-agent models in the analysis of large dynamic economies. 

Various methods have been proposed to address the limitations of rational expectations, yet as \cite{Moll2024} emphasizes, developing a viable alternative remains a fundamental challenge. The key question is how to construct tractable models of expectations formation that replace rational expectations while ensuring internal consistency, computational feasibility, partial immunity to Lucas Critique, and alignment with the behavioral literature. 

This paper introduces a novel equilibrium notion for solving large dynamic economies with aggregate shocks that balances tractability and internal model consistency while aligning with key ideas from the computational and behavioral literature. 
In our equilibrium notion, the $N$-Bounded Foresight Equilibrium (N-BFE), agents make optimization decisions over an infinite horizon but form expectations about key economic variables only for the next $N$ periods. Beyond this horizon, they assume these variables remain constant and rely on predetermined values to compute the continuation payoff from that point onward.  N-BFE is characterized by three conditions. First, agents' policy functions must be optimal given their foresight horizon and information they have. Second, the distribution of agents' states must evolve consistently with their decisions and aggregate shocks. Third, the agents' information must be internally consistent, meaning agents' predictions about future economic variables  align with the consistent evolution of the distribution of agents' states. As $N$ increases, N-BFE becomes closer to the rational expectations equilibrium in the sense that agents incorporate more forward-looking information into their decisions. On the other hand, when $N=0$ and there are no aggregate shocks, the 0-BFE resembles a stationary equilibrium notion.

While our approach departs from full rational expectations, N-BFE preserves the key principle that agents are forward-looking and make optimal decisions based on available information. However, instead of requiring agents to forecast the entire equilibrium path of key economic variables, such as market prices or total production, we impose a bounded foresight horizon, limiting the depth of their forward-looking predictions. Thus, beliefs about key economic variables respond to policy changes over a horizon of 
$N$ periods, providing partial immunity to the Lucas Critique by allowing expectations to adjust within the foresight window.  More formally, our foresight framework decomposes the agents' dynamic optimization problem into two stages: a finite-horizon segment, where agents have forecasts on the evolution of key economic variables, and an infinite-horizon stationary continuation, where these variables are assumed to remain constant. By restricting the depth of the foresight tree, this approach significantly reduces computational complexity.  Importantly, unlike traditional approaches that deviate from rational expectations which require the modeler to explicitly specify agents' subjective beliefs, our framework allows agents' forecasting errors to emerge endogenously, depending on their foresight depth and the actual complexity of the model.

A key motivation for our approach is computational considerations. Given the intractability of solving full rational expectations equilibrium in large heterogeneous-agent models, it is natural to impose a bounded depth on forward-looking decision-making, as the primary source of complexity in these models stems from the need to solve for distant future equilibria market prices. In this sense, our bounded foresight framework draws a  direct parallel to reinforcement learning (RL) methods such as the $N$-step lookahead policy \citep{bertsekas2019reinforcement}, particularly in systems that combine online optimization or search with offline function approximation. For example in RL systems such as  AlphaGo \citep{silver2016mastering}, an online agent, such as Monte Carlo Tree Search, performs real-time exploration of future states but is computationally constrained, much like how agents in our model have foresight over only $N$ future periods. Beyond this foresight horizon, both the RL agent and our economic agents rely on an offline valuation. In RL, this could be pre-trained deep neural networks that provide value estimates, while in our framework, it corresponds to simplified predictions, such as assuming constant economic variables or extrapolating from past trends. Our equilibrium notion extends the $N$-step lookahead idea from dynamic optimization to large dynamic economies by formally defining the internal consistency of equilibrium within a bounded depth in these settings.  

A second motivation for our approach is the substantial reduction in agents’ cognitive load while introducing behavioral properties that are arguably more realistic than full rational expectations. The rational expectations framework imposes an unrealistic informational and computational burden on agents, requiring them to form expectations about all future equilibrium prices. This assumption lacks plausibility, as real-world agents lack both the cognitive capacity and the necessary information to perform such high-dimensional forecasting. Agents in our bounded foresight model exhibit two key behavioral patterns that emerge endogenously. 

 First, agents display a form of time inconsistency \citep{laibson1997}, as their decisions evolve dynamically with their rolling $N$-period foresight window. Rather than committing to a single intertemporal policy, agents re-optimize their strategy in each period under the assumption that the economic variables of interest beyond their foresight horizon remain fixed. As the economy progresses, they continuously update their strategy based on their newly projected foresight. The time inconsistency in our model does not stem from present bias but rather from agents' inability to fully anticipate all future variables. As a result, there is little reason to expect their decisions to be time-invariant (see \cite{halevy2015time} for a detailed discussion on different types of time inconsistency).
Second, our framework naturally incorporates limited attention \citep{Gabaix2024}, as agents only forecast economic variables for the next $N$ periods and assume that these variables stabilize beyond that horizon. In this sense, agents do not have ability to monitor key economic variables over an infinite-
horizon. Together, these biases suggest that our approach provides a behaviorally preferred alternative to the rational expectations framework in complicated dynamic settings.

A third motivation for our framework is its natural compatibility with other computational methods to solve the dynamic optimization problem of the agent. For example, the bounded foresight equilibrium we propose can be seamlessly incorporated into RL-based approaches, where agents approximate value functions rather than solving the full dynamic programming problem explicitly. For instance, agents in our model can use function approximation techniques such as neural network-based value function learning or policy gradient methods to optimize their decision rules, making the approach scalable to high-dimensional settings. Moreover, our framework can be integrated with the popular moment-based approximations for solving large dynamic models with heterogeneous agents and aggregate shocks.

Moment-based methods preserve the infinite-horizon tracking of economic variables but simplify computation by summarizing the agents' distribution with a few moments rather than the full distribution. \cite{krusell1998income} focused on the first moment, while subsequent models have extended this approach to additional moments \citep{weintraub2010computational}. While computationally tractable, these approximations rely on the assumption that a small set of moments sufficiently captures the evolution of state variables. If economic dynamics exhibit richer heterogeneity, the approximation may become inaccurate. Additionally, this approach assumes that agents can approximate an infinite horizon of state variable evolution, which remains computationally demanding, and the idea that decision-makers use moments to forecast prices has faced criticism \citep{Moll2024}.  There are many variants that extend and improve the moment-based methods. For example, 
\cite{broer2022possibility} 
proposes a model where households optimally choose their information. \cite{fernandez2023financial} among others provide methodological innovations to improve the computation of equilibrium. Other popular methods to solve large dynamic economies include Taylor approximation techniques and ``MIT shocks", e.g., \cite{ahn2018inequality}, \cite{kaplan2018monetary},   \cite{boppart2018exploiting}, \cite{auclert2021using}, and \cite{bilal2023solving}.  See \cite{Moll2024} for a survey on these methods.

We introduce our equilibrium notion in a general setting where each agent has an individual state (such as the agent's current wealth or the firm's current capacity) and there is an aggregate state that influences all the agents. A prominent example  of such a model is Krusell-Smith type models \citep{krusell1998income} where each agent has an individual state that corresponds to the agent's wealth and the aggregate state influences the productivity of the economy. Another example is Ericson-Pakes type dynamic  models \citep{ericson1995markov} with many firms \citep{weintraub2008markov} where each firm has an individual state that corresponds to the firm's capacity or product's quality and the aggregate state influences the total demand for the firms' products.  
In Section \ref{Section:Model} we describe our model and equilibrium concept formally.  We prove the existence of equilibrium under mild continuity and compactness assumptions. 

Using our foresight framework, we provide a numerical comparative statics analysis that demonstrates an economically significant channel arising from variations in the degree of foresight. We focus on classical models in heterogeneous-agent macroeconomics and dynamic competition, including Krusell-Smith-type economies and capacity competition models. Across a  range of simulations, we find that a substantial variation in endogenous equilibrium variables is driven by foresight. This effect is distinct from mechanisms related to risk aversion or precautionary savings behavior, for example. Instead, it originates from a feedback loop between individual agent decision-making, which forecasts equilibrium variables, and the fact that equilibrium variables simultaneously reflect individual agents' decisions. We illustrate that as agents' capability to anticipate future equilibrium variables increases, they incorporate a higher degree of non-stationarity into their decisions. This, in turn, induces greater non-stationarity in the equilibrium  state variables.

Furthermore, we measure the endogenous degree of forecast errors across different foresight horizons. These errors, defined as the differences between forecasted and realized economic variables, arise endogenously in our bounded foresight model and depend on the complexity of the model under consideration.
Our analysis reveals that forecast errors increase as agents’ foresight horizons expand, a seemingly counterintuitive result. This occurs because greater foresight leads agents to internalize more variability in their decisions, which feeds back into both realized and forecasted equilibrium variables. As agents anticipate more fluctuations, they induce greater variability in the economy itself, making accurate forecasts more challenging. However, as the foresight horizon grows without bound and approaches the rational expectations benchmark, forecast errors necessarily decrease at some horizon.

Before presenting our model formally, we now provide an illustrative example of our foresight framework in the a Krusell-Smith type economy.

\section{Foresight in Krusell-Smith Economy}

In the Krusell-Smith economy, agents face an exogenous binary aggregate shock, $\tilde{z}_t$, which delineates either an expansionary or recessionary aggregate state. Within heterogeneous agent models featuring aggregate shocks, such as the Krusell-Smith model, these exogenous aggregate states are crucial in determining the trajectory of other endogenous economic variables that are established in equilibrium such as the interest rate. This process is depicted in Figure \ref{Fig:Inf} as a binary tree, where endogenous variables in the Krusell-Smith economy, like the wealth distribution, $\tilde{s}_t$, and the interest rate, $r_t$, evolve based on the sequence of the aggregate shock $z_t$.

 We now describe in detail the cases of infinite foresight and one foresight. 

\subsection{Infinite Foresight}

In the case infinite foresight, agents possess the capability to track every possible  pair of $(\tilde{s},r)$ across any future sequence of exogenous aggregate states. We illustrate this in the tree in Figure \ref{Fig:Inf}. At period $k$, the agent has a foresight of the interest rate $r_{k}$ along any string of aggregate states from period $1$ to period $k$, $(z_1,\ldots,z_k)$. 
 A typical equilibrium definition in this setting follows rational expectations and requires that interest rates are fully consistent with the model so  the interest rates are induced by the agents' dynamic savings decisions and the evolution of the wealth distribution in the economy. In this equilibrium notion agents ``correctly" predict expected future economic variables.

\begin{figure}[h]
\centering
\begin{tikzpicture}
     \draw (-5,0) -- (-3,1) ;
     \draw (-5,0) -- (-3,-1) ;
     \draw (-1,2) -- (1,3) ;
     \draw (-1,-2) -- (1,-3) ;
     \draw (-1,2) -- (1,1) ;
     \draw (-1,-2) -- (1,-1) ;
     \node at (-6.5,0) {$(\tilde{s}_{z_1},r_{z_1};z_1)$};
     \node at (-2.5,1.5) {$(\tilde{s}_{z_1z_1},r_{z_1z_1};z_1z_1)$};
     \node at (-2.5,-1.5) {$(\tilde{s}_{z_1z_2},r_{z_1z_2};z_1z_2)$};
     \node at (1.5,3.25) {$(\tilde{s}_{z_1z_1z_1},r_{z_1z_1z_1};z_1z_1z_1)$};
     \node at (1.5,-.75) {$(\tilde{s}_{z_1z_2z_1},r_{z_1z_2z_1};z_1z_2z_1)$};
     \node at (1.5,.75) {$(\tilde{s}_{z_1z_1z_2},r_{z_1z_1z_2};z_1z_1z_2)$};
     \node at (1.5,-3.25) {$(\tilde{s}_{z_1z_2z_2},r_{z_1z_2z_2};z_1z_2z_2)$};
\end{tikzpicture}
\caption{The figure illustrates agent ability in the infinite foresight case to have a forecast on economic variables along any possible path. In the case of binary aggregate states, the tree illustrates all possible future aggregate state histories for two periods starting at state $z_1$.} 
\label{Fig:Inf}
\end{figure}
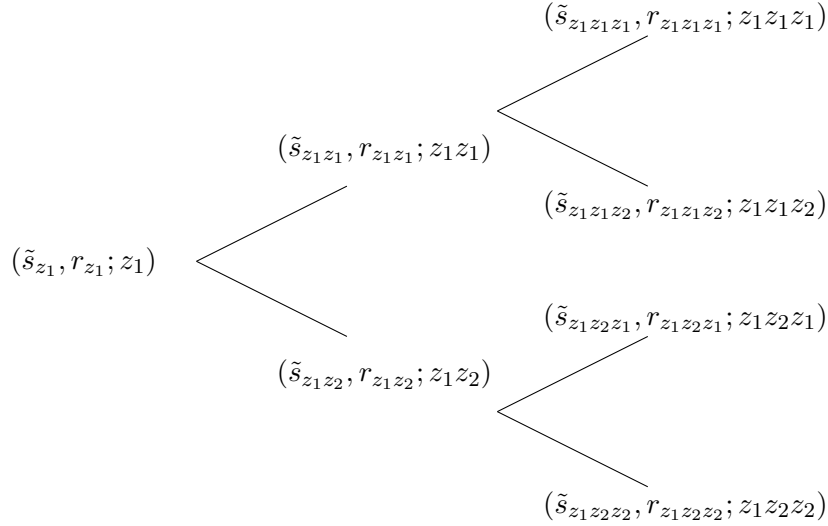

\subsection{One Foresight}

In the 1-bounded foresight equilibrium (1-BFE), agents take into account the possibility of different economic variables in the next period. For example, suppose the economy might be in either a ``good” or ``bad” state in period 1. Agents then forecast two different values for key economic variables (such as the interest rates in Krusell-Smith models) corresponding to these aggregate states. However, beyond period 1, they use predetermined continuation payoffs to evaluate their payoffs from that point onward, assuming the interest rate remains constant. Figure \ref{fig:OneForesight} illustrates agents' foresight: The tree branches at period 1 into two nodes, each representing a different forecast for the interest rate for different aggregate states. However, beyond period 1, agents maintain the same forecast for the interest rate at all future depths of the tree. This captures the idea that agents are forward-looking in the short run while simplifying their long-run outlook by ``freezing” expectations on economic variables after one period.

In our equilibrium notion, economic variables (e.g., interest rates) must be internally consistent with the model. This means that the interest rates anticipated by agents must align with the evolution of the wealth distribution and the agents' savings policy functions. 

Dynamic inconsistency arises because agents' foresight horizon shifts forward as time progresses, leading them to re-optimize their decisions based on newly available information. Once agents move forward one period, their foresight window also shifts, incorporating one additional period into their explicit planning while recalculating the continuation value. This rolling foresight structure creates a discrepancy between the decisions made in the past and the decisions that agents would make when they arrive at that future state, leading to time-inconsistent behavior. As discussed in the introduction, this decision-making in complicated settings aligns with many natural RL methods and the behavioral literature in economics.

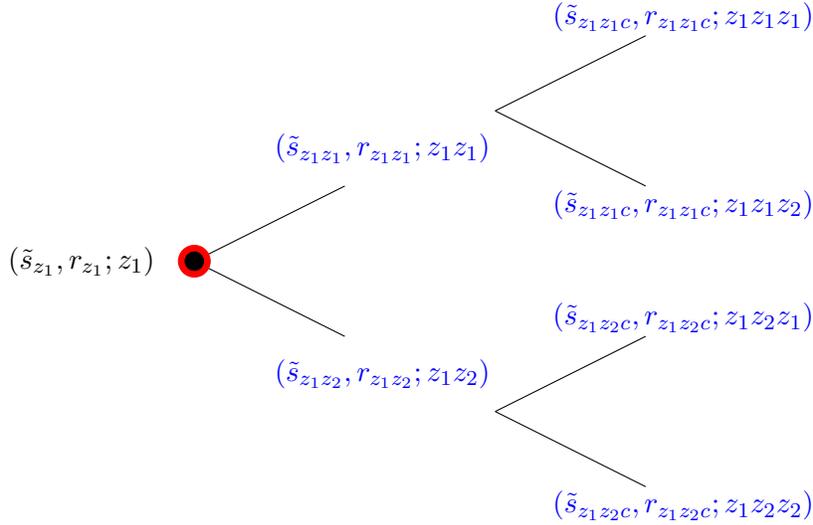
\begin{figure}[h]
\centering
\begin{tikzpicture}
     \draw (-5,0) -- (-3,1) ;
     \draw (-5,0) -- (-3,-1) ;
     \draw (-1,2) -- (1,3) ;
     \draw (-1,-2) -- (1,-3) ;
     \draw (-1,2) -- (1,1) ;
     \draw (-1,-2) -- (1,-1) ;
     \node[circle,draw=red, fill=red, inner sep=0pt,minimum size=12pt] at (-5,0) {};
     \node[circle,draw=black, fill=black, inner sep=0pt,minimum size=7pt] at (-5,0) {};
     \node at (-6.5,0) {$(\tilde{s}_{z_1},r_{z_1};z_1)$};
     \node[text=blue] at (-2.5,1.5) {$(\tilde{s}_{z_1z_1},r_{z_1z_1};z_1z_1)$};
     \node[text=blue] at (-2.5,-1.5) {$(\tilde{s}_{z_1z_2},r_{z_1z_2};z_1z_2)$};
     \node[text=blue] at (1.5,3.25) {$(\tilde{s}_{z_1z_1c},r_{z_1z_1c};z_1z_1z_1)$};
     \node[text=blue] at (1.5,-.75) {$(\tilde{s}_{z_1z_2c},r_{z_1z_2c};z_1z_2z_1)$};
     \node[text=blue] at (1.5,.75) {$(\tilde{s}_{z_1z_1c},r_{z_1z_1c};z_1z_1z_2)$};
     \node[text=blue] at (1.5,-3.25) {$(\tilde{s}_{z_1z_2c},r_{z_1z_2c};z_1z_2z_2)$};
\end{tikzpicture}
\caption{The figure illustrates the one-period foresight structure of agents. Starting from an aggregate state \( z_1 \), the tree branches into two nodes in period 1, each corresponding to a different possible aggregate state.  Agents forecast the interest rate for the next period but assume that from period 2 onward, is follows a constant value. Note that this value does not  necessarily equal the last realized interest rate but can be a function of the history, such as an average of past rates. The notation \( r_{z_1z_1c} \) and \( r_{z_1z_2c} \) indicates that the agent assumes these future interest rates remain fixed for all subsequent periods, conditional on the realized history up to period 1.}
\label{fig:OneForesight}
\end{figure}

\section{Model} \label{Section:Model}

This section presents the large dynamic economic model with aggregate shocks that we study in this paper. Many important models in the economic literature, ranging from the heterogeneous agent macro models to dynamic competition models, fall within our framework. We provide specific examples in Section \ref{Sec:application}, including quality ladder models and Krusell-Smith type models. We now formally present the model.

\textit{Time.} We study a discrete time setting. We index time periods by $t =1 ,2 ,\ldots $.

\textit{Agents.} There is a continuum of ex-ante identical agents of measure $1$. We use $i$ to denote a particular agent.

\textit{States.} There are individual  states for each agent and an aggregate state that is common to all agents. The individual state of agent $i$ at time $t$ is denoted by $x_{i ,t} \in X$ where $X$ is a complete separable metric space. The aggregate state at time $t$ is denoted by $z_{t} \in Z$ where $Z$ is a finite set. We let $\mathcal{P} (X)$ be the set of all probability measures on $X$ and $\mathcal{B}(X)$ be the Borel sigma-algebra on $X$. We denote by
$s_{t} \in \mathcal{P} (X)$ the probability measure that describes the distribution of agents' states at time $t$. We refer to $s_{t}$ as the \emph{population state} of time $t$.

\textit{Actions.} The action taken by agent $i$ at time $t$ is denoted by $a_{i ,t} \in A$ where $A \subseteq \mathbb{R}^{q}$. The set of feasible actions for an agent in state $x$ is given by a set $\Gamma  (x) \subseteq A$ .

\textit{States' dynamics.} The individual state of an agent evolves in a Markov fashion. 
If agent's $i$'s state at time $t -1$ is $x_{i ,t -1}$, the agent takes an action $a_{i ,t -1}$ at time $t -1$, the population state at time $t$ is $s_{t-1}$, the aggregate state at time $t$ is $z_{t}$, and $\zeta _{i ,t}$ is agent $i$'s realized idiosyncratic random shock at time $t$, then agent $i$'s next period's state is given by\footnote{For notational simplicity, we assume that the evolution dynamics of the individual state $x_{t}$ depend on the aggregate state $z_{t}$, rather than $z_{t-1}$, as is the case in our applications. However, the model can be easily extended to incorporate dependence on $z_{t-1}$ as well.
}
\begin{equation*}x_{i ,t} =w (x_{i ,t -1}, z_{t} ,a_{i ,t -1} ,s_{t-1} ,\zeta _{i ,t}).
\end{equation*}
We assume that $\zeta _{i,t} $ are independent and identically distributed random variables across time and agents that take values on a compact separable metric space $E$ and  have a law $q$. The aggregate states evolve in a Markovian fashion according to a finite Markov chain with a transition matrix $P$. For exposition simplicity, we assume in this section that the evolution of the aggregate states is independent from the evolution of the idiosyncratic shocks $\zeta _{i ,t}$.\footnote{Our model can be easily generalized to the case that the idiosyncratic shocks $\zeta$ depends on the realized aggregate state (e.g., in \cite{krusell1998income} the idiosyncratic labor shocks depend on the realized aggregate productivity shock).}
We call $w : X  \times Z \times A \times \mathcal{P}(X) \times E \rightarrow X$ the transition function.

\textit{Payoff.} In a given time period, if the aggregate state is $z$, the state of agent $i$ is $x_{i}$, the population state is $s$, and the action taken
by agent $i$ is $a_{i}$, then the single period payoff for agent $i$ is $\pi (x_{i} ,z ,a_{i} ,s)$. The agents discount their future payoff
by a discount factor $0 <\beta  <1$. Thus, a agent $i$'s infinite horizon payoff is given by: $\sum _{t =1}^{\infty }\beta ^{t-1} \pi  (x_{i ,t} , z_{t},a_{i ,t} ,s_{t})$.

As explained in the introduction, the key distinction between our model and the standard literature on large dynamic games with aggregate shocks lies in the agents' limited foresight and their deviation from the rational expectations framework. We now formally describe the information available for agents in our setting.

\textit{Information structure.} 
 Let $Z^{t}:=\underbrace{Z \times \ldots  \times Z}_{t~ \mathrm{t} \mathrm{i} \mathrm{m} \mathrm{e} \mathrm{s}}$  be the space of all finite aggregate shock histories of length $t$. Every history of aggregate shocks $z^{k} \in Z^{k}$ has an associated population state $s_{k}(z^{k})$ in period $k$. For a period $t \in \mathbb{N}$, an initial aggregate state $z_{t}$, and $N$ degrees of foresight, we define an information structure $\mathcal{I}(t,N,z_{t})$ by a vector of future population states:
\begin{equation*}\mathcal{I}(t,N,z_{t}) =(s_{t}(z_{t}) ,s_{t +1}(z^{t ,t +1}) ,\ldots  ,s_{t +N}(z^{t ,t +N}) ,\ldots )
\end{equation*} 
where $z^{t ,t +k} \in Z^{k +1}$ for $k =1 ,\ldots $ is a history of aggregate shocks of length $k +1$ with an initial state $z_{t}$ and for each history $z^{t ,t +k}$ with $k >N $ such that its projection to the first $N$ periods is exactly $z^{t ,t +N}$ we have $s_{t +k}
(z^{t ,t +k}) =s^{C}(z^{t,t+N})$ for some constant population state that may depend on the history $z^{t,t+N}$ but not on $k$. The vector $\mathcal{I}(t,N,z_{t})$ contains the information that an agent has about future population states. More precisely, it contains all the possible future population states in the next $N$-periods when the current aggregate shock is $z_{t}$ and the population states after period $N$ are assumed to be fixed.  An agent with $N$-degrees of foresight has the information structure $\mathcal{I}(t,N,z_{t})$ in period $t$, and thus, the agent has a prediction on the possible population states for every future history of length $k \leq N$. 
Note that for an information structure $\mathcal{I}(t,N,z_{t})$, the population states for any two histories $(z_{t}, \ldots, z_{t+\tau})$ and $(z_{t}, \ldots, z_{t+\tau'})$ such that $\tau' \geq \tau > N$ and the projection of $(z_{t}, \ldots, z_{t+\tau'})$ onto the first $t + \tau$ periods is $(z_{t}, \ldots, z_{t+\tau})$, are both equal to the constant population state $s^{C}(z_{t}, \ldots, z_{t+\tau})$.

The constant population state $s^{C}(z^{t,t+N})$ can be computed in various ways. For instance, it could represent the average population state along the given history sample path, or it might simply be the population state from the last period of the foresight window (as we assume in our simulations). Generally,  the specific choice of $s^{C}(z^{t,t+N})$ can depend on the concrete application studied and how agents perceive the stability of the population state beyond their foresight window.

 In our applications, the population states determine key economic variables such as the interest rate in the Krusell-Smith economy  or the total production in the dynamic capacity competition models. 
 Hence, the interpretation of $N$-degrees of foresight for an agent is that the agent has a forecast about the possible economic variables for the next $N$ periods and assumes they remain constant thereafter. 
For concreteness, assume that the agent has foresight over an exchange rate. If the agent has $0$-degrees of foresight and there are two possible aggregate states (growth or recession), they assume that the current exchange rate will remain fixed forever when optimizing their decisions. This is despite recognizing that the economy may transition between different states in the next period, they simply do not know how to evaluate its impact on the exchange rate. In contrast, an agent with $1$-degree of foresight anticipates that the exchange rate may take different values in the next period depending on whether the economy experiences growth or recession. 

For ease of notation, we often drop the subscripts $i$ and $t$ and denote a generic transition function by $w (x ,a ,s ,\zeta )$ and a generic payoff function by $\pi  (x ,a ,s)$.

For the rest of the section, we assume the following conditions on the primitives
of the model:

\begin{assumption} \label{Assumption:1}
(i) The payoff function $\pi $ is bounded and (jointly) continuous and the transition function $w$ is continuous.\footnote{
We endow $\mathcal{P} (X)$ with the weak topology. We
say that $s_{n} \in \mathcal{P} (X)$ converges weakly to $s \in \mathcal{P} (X)$ if for all bounded and continuous functions $f :X \rightarrow \mathbb{R}$ we have
\begin{equation*}\underset{n \rightarrow \infty }{\lim }\int _{X}f (x) s_{n} (d x) =\int _{X}f (x) s (d x).
\end{equation*}
}

(ii) The state space $X$ is compact.

(iii) The correspondence $\Gamma  :X \rightarrow 2^{A}$ is compact-valued and continuous.\footnote{\samepage
By continuous we mean both upper hemicontinuous and lower hemicontinuous (see Chapter 17 in \cite{aliprantis2006infinite} for the definitions). }

\end{assumption}

\begin{remark}
For notational simplicity, we assume that agents are ex-ante homogeneous. However, our framework naturally extends to settings with ex-ante heterogeneity, where each agent is assigned a fixed type at the initial period, drawn from some distribution. In this case, both the payoff and transition functions may depend on the agent's type. The formulation and results in this paper remain valid in this more general setting. To keep the exposition concise, we omit these details but refer to \cite{light2022mean} for a discussion of such extensions.
\end{remark}

\subsection{The Agents' Optimization Problem} \label{Sec: agent problem}

In this section we describe the single-agent dynamic optimization problem with $N$ degrees of foresight and starting from some period $t$. 

 Let $X^{k}:=\underbrace{X \times \ldots  \times X}_{k~ \mathrm{t} \mathrm{i} \mathrm{m} \mathrm{e} \mathrm{s}}$. Given an information structure in period $t$ of an agent with $N$ degrees of foresight, a nonrandomized policy $\sigma $ starting from period $t$ is a sequence of (Borel) measurable functions $(\sigma _{t} ,\sigma _{t +1} ,\ldots  ,)$ such that
$\sigma _{k} :Z ^{k -t +1} \times X \rightarrow A$ and $\sigma _{k} (z_{t} ,\ldots z_{k}, ,x_{k}) \in \Gamma  (x_{k})$ for $k =t ,t +1 ,t +2 ,\ldots $. That is, a policy prescribes a feasible action for any possible finite sequence of aggregate states and the current individual state.

 When an agent uses a policy $\sigma $, the agent's information structure is $\mathcal{I}(t,N,z_{t})$, the aggregate state is $z_{t} \in Z$ and the agent's  state is $x_{t} \in X$, then the agent's expected present discounted payoff is
\begin{equation*}V^{\mathcal{I}(t,N,z_{t})}_{\sigma }( x_{t}, z_{t}, s_{t} ) =\mathbb{E}_{\sigma } \left ( \sum _{ t' =t }^{\infty }\beta ^{t'  -t} \pi  (x_{t'} ,z _{t'} ,a _{t'} ,s_{t'}) \right )
\end{equation*}
where $\mathbb{E}_{\sigma }$ denotes the expectation with respect to the probability measure that is generated by the policy $\sigma$.\protect\footnote{ The probability measure is uniquely defined (see \cite{bertsekas1996stochastic} for details).
} 

More generally, we let $V^{\mathcal{I}(t,N,z_{t})} _{\sigma} (x_{\tau} , z_{\tau}, s_{\tau }(z_{t} ,\ldots  ,z_{\tau }) )$ to be the agent's expected discounted payoffs from period $\tau$ when the agent's state is $x$ and the history of aggregate states from period $t$ to $\tau$ is ($z_{t} ,\ldots  ,z_{\tau })$.
Denote
\begin{equation*}
V^{\mathcal{I}(t,N,z_{t})} (  x_{\tau} , z_{\tau}, s_{\tau }(z_{t} ,\ldots  ,z_{\tau })) =\underset{\sigma }{\sup } \ V_{\sigma }^{\mathcal{I}(t,N,z_{t})}(x_{\tau} , z_{\tau}, s_{\tau }(z_{t} ,\ldots  ,z_{\tau })).
\end{equation*}
We call $V$ the \emph{value function }and a policy $\sigma $ attaining it \emph{optimal}.

Using  standard dynamic programming arguments, we can show that the dynamic optimization problem for an agent with $N$-degrees of foresight, $N \in \mathbb{N}_{0}$ can be solved recursively, and a version of the Bellman's principle of optimality holds. The proof is deferred to the Appendix. 

\begin{proposition} \label{Prop:ValueF}
Suppose that Assumption \ref{Assumption:1} holds.   The value function satisfies 
\begin{align*}
& V^{\mathcal{I}(t,N,z_{t})}\left(x_{\tau} , z_{\tau}, s_{\tau}(z_{t} ,\ldots  ,z_{\tau }) \right ) =\max _{a \in \Gamma (x_{\tau})}\pi (x_{\tau } ,z_{\tau} ,a ,s_{\tau }(z_{t} ,\ldots  ,z_{\tau })) \\ + & \beta\int _{E} \sum _{z_{i} \in Z}P(z_{\tau
 },z_{i}) V^{\mathcal{I}(t,N,z_{t})} \left (w(x_{\tau}, z_{i},a ,s_{\tau }(z_{t} ,\ldots  ,z_{\tau }) ,\zeta ) , z_{i} , s_{\tau +1}(z_{t} ,\ldots  ,z_{\tau },z_{i})  \right )q(d
\zeta) \end{align*}
for every string of aggregate states $(z_{t}
,\ldots  ,z_{\tau })$ and state $x_{\tau} \in X$. In addition, the optimal policy correspondence $G^{\mathcal{I}(t,N,z_{t})}\left(x_{\tau} , z_{\tau}, s_{\tau}(z_{t} ,\ldots  ,z_{\tau }) \right )$ consists of the actions that achieve the maximum in the equation above. 
\end{proposition}

\subsection{N-Bounded Foresight Equilibrium}

We now provide a general equilibrium notion that builds on the bounded foresight of agents. This equilibrium notion deviates from both the rational expectations equilibrium and the stationary equilibrium. Importantly, the equilibrium notion we introduce is computationally tractable for relatively small $N$ while for large $N$ it suffers from the curse of dimensionality as the rational expectations equilibrium.

 To simplify notation, we sometimes write $s(z_{t} ,\ldots  ,z_{\tau })$ instead of $s_{\tau}(z_{t} ,\ldots  ,z_{\tau }) $. 
 For any history $z^{t ,t +k} \in Z^{k +1}$ that starts with $z_{t} \in Z$, let $s(z^{t ,t
 +k} ,z_{i} ) ( \cdot )$ be the probability measure on $\mathcal{P}(X)$ that describes the population state given the history of aggregate states $(z^{t ,t +k} ,z_{i})$ and the  agents' policy functions $g^{\mathcal{I}(t,N,z_{t})}$. That is, $s(z^{t ,t
 +k} ,z_{i} )$ is defined recursively by 
 \begin{align} \label{Eq:Dist-dynamics}
 s(z^{t ,t +k} ,z_{i}) (D) &=\int _{X}
 \int _{E} 
 & 1_{D} \left (w(x,z_{i }, g^{\mathcal{I}(t,N,z_{t})}(x ,z_{t+k}, ,s(z^{t ,t +k} ) ), s(z^{t ,t +k} ), \zeta ) \right ) q(d\zeta)s(z^{t ,t +k} )(dx) 
\end{align}
for any $D \in \mathcal{B}(X)$ where $1_{D}$ is the indicator function.
Equation (\ref{Eq:Dist-dynamics}) describes the dynamics of the population states over time given the agents' policy functions. 

We now define the $N$-Bounded Foresight Equilibrium ($N$-BFE). This equilibrium is non-stationary and depends on the initial period $t$, as agents' optimal policies vary with time. In such an equilibrium, agents' policies must be optimal given their information, and population states must be consistent with both the agents' policies and their current information structure. Specifically, the population states and policies  evolve in accordance with the predictions from the current information structures, where foresight is limited to $N$ periods into the future.

\begin{definition}
Fix a period $t$, an initial aggregate shock $z_t$, an initial population state $s(z_t)(\cdot)$, and $N \in \mathbb{N}_0$. An $N$-Bounded Foresight Equilibrium ($N$-BFE) is a collection of an information structure $\mathcal{I}(t,N,z_{t})$, policy functions $g$, and population states such that:
\begin{enumerate}
    \item \textbf{Optimality:}  
    Given the information structure $\mathcal{I}(t,N,z_{t})$, agents’ decisions are optimal. That is, $g^{\mathcal{I}(t,N,z_{t})} \in G^{\mathcal{I}(t,N,z_{t})}$ solves the agents' dynamic optimization problem as described in Section \ref{Sec: agent problem}.
    
    \item \textbf{Population State Consistency:}  
    For any history $z^{t ,t+k}$ with $k \leq N$ that starts with $z_{t}$, the population state $s(z^{t ,t+k}) (\cdot)$ is consistent with the agents’ optimal decisions and evolves according to Equation \eqref{Eq:Dist-dynamics}.
    
    \item \textbf{Information Structure Consistency:}  
    The information structure $\mathcal{I}(t,N,z_{t})$ is consistent with the agents' beliefs about future population states. Specifically, for every history $z^{t ,t+k}$ with $k \leq N$ that starts with $z_{t}$, the population state $s(z^{t ,t+k}) $ predicted by the information structure $\mathcal{I}(t,N,z_{t})$ is equal to the population states described in part 2 of the definition. 
\end{enumerate}
\end{definition}

In the case that $N=0$, if there are no aggregate shocks, and the population state $s$ is invariant,  we retrieve the standard stationary equilibrium, a popular solution concept in a wide range of large dynamic economies and stationary mean field models (e.g., \cite{aiyagari1994uninsured},  \cite{acemoglu2015robust}, and \cite{light2022mean}). On the other hand, if $N=\infty$ so the agents' foresight is not bounded, we retrieve the rational expectations recursive equilibrium (e.g., \cite{krusell1998income}). 

We also note that, in this model, we assume each agent has $N$ periods of foresight. However, the model can readily be extended to cases where agents have different foresight horizons. Specifically, let $N_i$ denote the foresight horizon of agent $i$. In this case, each agent $i$ optimizes their decisions given their own information structure,  the population state needs to be consistent with the policies of all agents, and the equilibrium definition would need to reflect that each agent's foresight is consistent. 

We now establish that an N-BFE exists under Assumption \ref{Assumption:1} and the assumption that there is an optimal policy that is continuous. The proof is deferred to the Appendix. 
\begin{theorem} \label{Theorem existence} Suppose that Assumption \ref{Assumption:1} holds and that there is a continuous selection from the optimal policy correspondence $g^{\mathcal{I}(t,N,z_{t})} \in G^{\mathcal{I}(t,N,z_{t})}$. Then there exists a N-BFE.
\end{theorem}
From the proof of Proposition \ref{Prop:ValueF}, if $g^{\mathcal{I}(t,N,z_{t})}$ is single-valued, then it is continuous, so existence holds in this case. Standard dynamic programming arguments can be used to establish single-valuedness, particularly when the payoff function satisfies certain concavity conditions (for example, see the analysis of concave and continuous dynamic programming in \cite{light2024principle}). Continuity also typically holds in the case of finite state spaces, which applies to many  applications of interest.

We note that our equilibrium notion can also be extended to accommodate more general settings where additional equilibrium parameters are not determined as a one-to-one function of the population state but instead must satisfy a consistency condition with it.\footnote{For example, in \cite{huggett1993risk}, the equilibrium interest rate is determined by the requirement that aggregate savings equal zero.} Our equilibrium definition can be modified to allow such consistency constraints rather than explicit functional relationships. We omit the details for brevity.

\section{Applications} 
\label{Sec:application}

In this section, we present our applications. To illustrate our results, we focus on two classical applications in economics from different domains. First, we consider a Krusell-Smith type model from the heterogeneous agent macroeconomics literature. Then, we study quality ladder and capacity competition models from the industrial organization literature. In both settings we compute bounded foresight equilibria for different horizons and show how foresight influences key economics variables such as interest rates and total  production.

\subsection{Heterogeneous Agent Models with Aggregate Shocks}

Heterogeneous agent models with aggregate shocks are widely studied in macroeconomics. 
For concreteness, in this section, we focus on the Krusell-Smith type economy \citep{krusell1998income}, which studies  general equilibrium in the presence of production and aggregate shocks. In this model, each agent's consumption and savings decision-making is influenced by both their individual state (e.g., wealth) and the aggregate state of the economy (e.g., productivity or aggregate capital). Market clearing conditions determine the equilibrium prices (e.g., interest and wage rates). In our framework, agents have bounded foresight regarding the evolution of these equilibrium prices as described in Section \ref{Sec: agent problem} and we solve for bounded foresight equilibrium. We note that our bounded foresight setting can be easily adopted to other recent heterogeneous agent macro models (e.g., the various models studied in \cite{acemoglu2015robust} or HANK models \citep{kaplan2018monetary}). We now present the model formally.

\textit{States.} The state of agent $i$ at time $t$ is denoted by $x_{i ,t} =(x_{i ,t ,1} ,x_{i ,t ,2}) \in X_{1} \times X_{2} =X$ where $X_{1} \subseteq \mathbb{R}$, $X_{2} \subseteq \mathbb{R}$ and $x_{i ,t ,1}$ represents agent $i$'s savings and $x_{i,t,2}$ represents agent $i$'s labor productivity at period $t$.

\textit{Actions.} At each time $t$, agent $i$ chooses an action $a_{i ,t} \in \Gamma  (x_{i ,t}) \subseteq \mathbb{R}$ where 
$$\Gamma  (x_{i,t,1} ,x_{i,t,2})=[ -\underline{b} ,R_{t} x_{i,t,1} +w_{t} x_{i,t,2}],$$ 
$R_{t}$ and $w_{t}$ are the rate of return and wage rate, respectively. That is, the agents' savings are restricted to be lower than their cash-on-hand $R_{t} x_{i,t,1} +w_{t} x_{i,t,2}$ and higher than the exogenous  borrowing constraint $\underline{b} \geq 0$.

The wage rate and the interest rate are determined in general equilibrium. There is a representative firm with a production function $z_{t}F(K_{t},L_{t})$  where $L_{t}$ is the labor supply, $K_{t}$ is the capital, and $z_{t}$ is the aggregate shock in period $t$. We assume that $F$ is strictly concave, strictly increasing, homogeneous of degree one, and continuously differentiable. The optimality conditions for the firm's maximization problem yield\footnote{The firm's maximization problem is given by $\max_{K,N} F(K,N) - (R-1+\delta)K - wN$. For more details see, for example, \cite{acemoglu2015robust} and  \cite{Light2017}.} $z_{t}F_{k}(K_{t} ,N_{t}) =R_{t} +\delta  -1$ and $z_{t}F_{N}(K_{t} ,N_{t}) =w_{t}$ where $F_{i}$ denotes the partial derivative of $F$ with respect to $i =K ,N$ and $\delta  >0$ is the depreciation rate. Hence, normalizing labor to $1$ (or considering $K$ as the capital-to-labor ratio)\footnote{Since $F$ is homogeneous of degree one we have $F(K,1)=KF_{K}(K,1) + F_{N}(K,1)$.} $R_{t}=z_{t}f'(K_{t})-\delta +1$ and $w_{t} =z_{t}f (K_{t}) -z_{t}f^{ \prime } (K_{t}) K_{t}$ where $F(K_{t},1)=f(K_{t})$.  

 In equilibrium we have $\int _{X}x_{1} s_{t} (d (x_{1} ,x_{2})) =K_{t}$ where $s_{t}$ is the savings-labor productivities distribution at period $t$. That is, the aggregate supply of savings equals the total capital. 
We denote the aggregate savings by  $H (s_{t}) =\int _{X}x_{1} s_{t} (d (x_{1} ,x_{2}))$.

\textit{States' dynamics.}  In each period agents choose how much to save for future consumption and how much to consume in the current period. The labor productivity function $m$ determines the next period's labor productivity given the current labor productivity level. Formally, if agent $i$'s state at time $t -1$ is $x_{i ,t -1}$, agent $i$ takes an action $a_{i ,t -1}$ at time $t -1$, and $\zeta _{i ,t}$ is agent $i$'s realized idiosyncratic random shock at time $t$, then agent $i$'s state in the next period is given by
\begin{equation*}(x_{i ,t ,1} ,x_{i ,t ,2}) =(a_{i ,t -1} ,m (x_{i ,t -1 ,2} ,\zeta _{i ,t})),
\end{equation*}
where $m :X_{2} \times E \rightarrow X_{2}$ is a continuous function. 

\textit{Payoff.} The agents' utility from consumption is given by a utility function $u$ which is assumed to be strictly concave and strictly increasing. If the agents choose to save $a$ then their consumption in period $t$ is $R_{t} x_{t,1} +w _{t} x_{t,2} -a$. Thus, using the firm's optimization conditions, the payoff function is given by
\begin{equation*}\pi  (x ,z, a ,s) =u\left((zf^{ \prime } (H (s)) - \delta + 1)x_{1} + \left (z f (H (s)) -z f^{ \prime } (H (s)) H (s) \right ) x_{2} -a\right).
\end{equation*}

Assuming compactness of the state space, it is easy to establish that the optimal policy correspondence $G$ is single-valued and that Assumption \ref{Assumption:1} holds. Thus, the existence of an equilibrium follows from Theorem \ref{Theorem existence}. 

In the simulations, an $N$-period bounded foresight agent assumes that the constant population state after $N$-periods equals the last observed population state. That is, for each history $z^{t ,t +k}$ with $k >N -1$ such that its projection to the first $N$ periods is exactly $z^{t ,t +N}$ we have $s_{t +k}(z^{t ,t +k}) =s_{t+N}(z^{t,t+N})$. 

We compute the bounded foresight equilibrium with zero, one and two periods of foresight\footnote{We match the implementation as in https://github.com/econ-ark/HARK/ of the Krusell-Smith model that replicates the \citep{krusell1998income} simulation}. We assume that agents' utility functions are given by $u(c) = \log(c)$.

A key statistic we focus on is the standard deviation of the mean of aggregate capital. In the Krusell-Smith model, there is a direct one-to-one relationship between aggregate capital and the equilibrium interest rate, a key economic variable in the individual's optimization.

We simulate the model along a path of aggregate shocks for 500 periods. In table \ref{tab:KS-capital}, we present the standard deviation of mean capital along a single simulated path of aggregate shocks $(z_1,...,z_{500})$. Extending the foresight horizon from zero to just two periods shows that the variability in mean capital approaches to the levels observed in the moment-based approximation of the Krusell-Smith model \citep{krusell1998income}. 

The plot in figure \ref{fig:KS_sim}  shows a single simulated path of aggregate shocks and demonstrates the variation of mean equilibrium capital to aggregate shocks is greatest in the moment-based approximation of the Krusell-Smith model, and that each degree of additional foresight appears to increase that variation under the foresight framework. In particular, the plot highlights the response to consecutive high productivity and low productivity aggregate states in the model.

\begin{table}[h]
\caption{The table lists the average standard deviation of mean capital for each of the zero, one and two period foresight models in addition to the moment-based approximation of the Krusell-Smith model. That is for a single path of 1000 periods, we compute the standard deviation of mean population capital. We then do this over 1000 paths to compute the exoected variability. Relative to the greater variability in the moment-based Krusell-Smith model, we observe that the standard deviation in mean capital increases with the degree of foresight.} 
\label{tab:KS-capital}
  \begin{center}
    \begin{tabular}{c||c|c|c|c} 
       & \textbf{Zero} & \textbf{One} & \textbf{Two} & \textbf{Full Horizon (Moment-Based)}\\
      \hline
      $\sigma_{\bar{K}}$  & 0.265 & 0.324 & 0.377 & 0.483  \\
    \end{tabular}
  \end{center}
\end{table}

\begin{figure}[h]
\centering
\includegraphics[width=.6\textwidth]{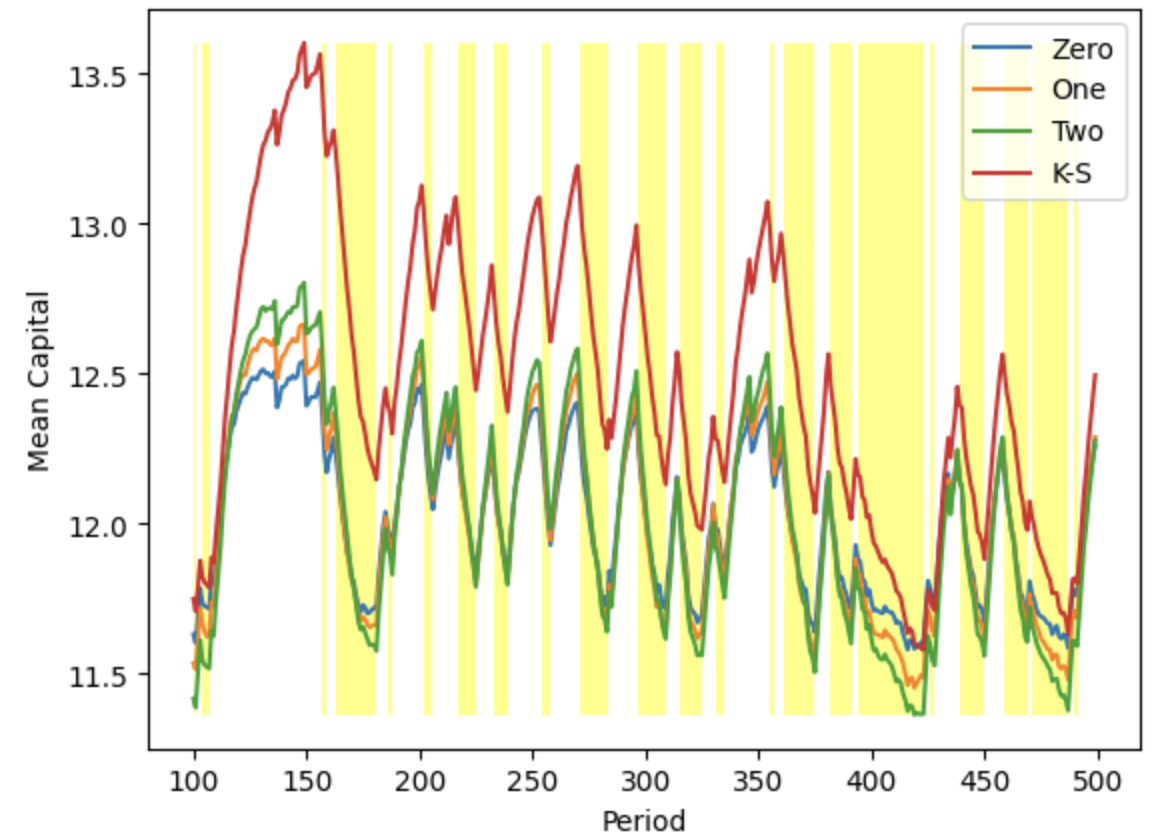}
\caption{This figure plots the mean equilibrium capital for zero, one and two foresight models along with the moment-based approximation of the Krusell-Smith model along a single aggregate history. Each of the models is faced with the same path of aggregate shocks. The yellow part indicates recessionary periods in the Krusell-Smith sense, i.e. the realization of a low aggregate state $z$. We  observe that the moment-based approximation of the  Krusell-Smith model, in red, has the greatest variation in mean capital to recessionary shocks, and that the zero, one and two foresight models have increasing variability. The green path illustrates that the peaks and troughs of the mean capital plot in the two foresight model surpass those of the zero foresight model in blue.}
\label{fig:KS_sim}
\end{figure}

We interpret the increased variation in the following way. When agents possess greater foresight, they adjust their savings decisions by incorporating the anticipated fluctuations in equilibrium capital, and consequently, in the interest rates into their optimization. In contrast, when foresight is limited, agents still optime savings decisions for the long run but effectively assume that the mean capital remains stationary beyond their  forecasting window. As a result, they underestimate the true variability of future prices, leading to a muted equilibrium response. In other words, extending foresight triggers a feedback mechanism whereby more realistic expectations of future fluctuations amplify the variability in equilibrium outcomes.

This result represents a distinct channel from typical sources of variability such as precautionary savings and risk aversion. The mere anticipation of future fluctuations, enabled by an extended foresight window, induces changes in savings decisions that amplify variability in key equilibrium variables.

\subsection{Dynamic Oligopoly Models with Aggregate Shocks}

In this section we study various dynamic models of competition  or dynamic oligopoly models that are widely studied in the economics literature. Although we study models with a large number of firms, we retain the term dynamic oligopoly models to remain consistent with previous literature, where models with many firms are often used to approximate oligopolistic behavior \citep{weintraub2008markov}. We show that these models naturally fit to our framework and present numerical comparative statics results regarding how different foresight horizons influence key market variables. 

We consider dynamic capacity models and dynamic quality ladder models, both of which have received significant attention in the literature. In these models, a firm's state corresponds to a key variable that influences its profits—such as its capacity or the quality of its product. Per-period profits are determined by a static competition game that depends on the heterogeneous firms' state variables and an aggregate shock shared by all firms. Firms take actions to improve their individual states over time. We now describe the models we consider.

\textit{States.} The state
of firm $i$ at time $t$ is denoted by $x_{i ,t} \in X$ where $X = \{0,\ldots,M\}$ is a discrete state space.  

\textit{Actions.} At each time $t$, firm $i$ invests $a_{i ,t} \in A =[0 ,\bar{a}]$ to improve its state. The investment changes the firm's state in a stochastic fashion. 

\textit{States' dynamics.} A firm's state evolves in a Markovian fashion. We adopt dynamics similar to those in \cite{pakes1994computing,adlakha2015equilibria,Ifrach2017}, which have been widely used in such models. 

In each time period, a firm's investment of $a$ is successful with probability $a/(1+a)$, in which case its state (e.g., quality or  level) increases by one level. The firm's state depreciates by one level with probability $\delta \in (0,1)$, independently in each time period. Thus, the firm's state decreases by one with probability $\delta/(1+a)$, increases by one with probability $(1-\delta)a/(1+a)$, and remains at the same level with probability $(1-\delta + \delta a)/(1+a)$.

\textit{Payoff.} The cost of a unit of investment is $d >0$. There is a single-period profit function $u (x,z ,s)$ that is derived from a static game. When a firm invests $a \in A$, then the firm's single-period payoff function is given by $\pi  (x,z ,a ,s) =u (x,z ,s)-da$.  

We now provide two classic examples of static competition models that are commonly studied in the literature and provide an explicit $u (x,z ,s)$. 

The first model is the  competition model based on \cite{besanko2004capacity} and \cite{besanko2010lumpy}. We consider an industry with homogeneous products, where each firm's state variable determines its production. If a firm's state is $x$, its production capacity is given by some increasing function $\bar{q}(x)$. Each period, firms take costly actions to increase their capacity for the following period. Additionally, firms compete in a capacity-constrained quantity-setting game each period. The inverse demand function is $P(Q)$, where $Q$ represents the total industry output. For simplicity, we assume that all firms have zero marginal costs.
Given the total quantity produced by its competitors, $Q_{-i}$, firm $i$ solves the following profit maximization problem: 
$
\underset{0 \leq q_{i} \leq \bar{q}(x_{i})}{\max} \ P(q_{i} + Q_{-i})q_{i}.
$

We could solve for the equilibrium of the capacity-constrained static quantity game played by the firms, and these static equilibrium actions would determine the single-period profits. However, for simplicity, we focus on the limiting regime with a large number of firms, where no individual firm has market power. In this case, firms take $Q$ as fixed and each firm produces at full capacity. The limiting profit function is then given by:
\[
u(x, z, s) =z P \left( \int_X \bar{q}(y) s(dy) \right) \bar{q}(x),
\]
where $z$ is an aggregate shock that reflects the state of the economy and is common to all firms.

Our second example is a classic quality ladder model, where individual states represent the quality of a firm's product (see, e.g., \cite{pakes1994computing} and \cite{ericson1995markov}). Consider a price competition under a logit demand system. There are $N$ consumers in the market. The utility of consumer $j$ from consuming the good produced by firm $i$ at period $t$ is given by
\begin{equation*}u_{i j t} =\theta _{1} \ln  (x_{it} +1) +\theta _{2} \ln  (Y-p_{it}) +v_{ijt},
\end{equation*}
where $\theta _{1}<1 ,\theta _{2} >0$, $p_{it}$ is the price of the good produced by firm $i$, $Y$ is the consumer's income, $x_{it}$ is the quality of the good produced by firm $i$, and $\{v_{i j t}\}_{i ,j ,t}$ are i.i.d Gumbel random variables that represent unobserved characteristics for each consumer-good pair. 

There are $m$ firms in the market and the marginal production cost is constant and the same across firms. There is a unique Nash equilibrium in pure strategies of the pricing
game (see \cite{caplin1991aggregation}). These equilibrium static prices determine the single-period profits. We focus on the limiting profit function that can be obtained from the asymptotic regime in which the number of consumers $N$ and the number of firms $m$ grow to infinity at the same rate. The limiting profit function is given by: 
\begin{equation*}u (x,z ,s) =z \frac{\tilde{c} (x +1)^{\theta _{1}}}{\int _{X}(y +1)^{\theta _{1}} s (d y)}
\end{equation*} (see \cite{light2022mean}). The aggregate state $z$ is common to all firms and represent the state of the economy, and
$\tilde{c}$ is a constant that depends on the limiting equilibrium price, the marginal production cost, the consumer's income, and $\theta _{2}$.

We assume, as in the Krusell-Smith model, that  in the simulations an $N$-period bounded foresight agent assumes that the constant population state after $N$-periods equals the last observed population state. That is, for each history $z^{t ,t +k}$ with $k >N -1$ such that its projection to the first $N$ periods is exactly $z^{t ,t +N}$ we have $s_{t +k}(z^{t ,t +k}) =s_{t+N}(z^{t,t+N})$.

We now provide comparative statics analysis for the capacity competition model. We assume a linear demand structure where $P(Q) = \frac{e}{f} - \frac{Q}{mf}$. In the simulation, we also assume the same set of idiosyncratic and aggregate shocks for each of the various foresight paths in order to directly compare the impact of greater foresight. Finally, we simulate 1000 firms in each period in order to avoid any dynamics influenced by a small set of firms.

\begin{table}[h!]
\caption{Simulation Parameters - Capacity Constraint Model. Parameters are similar to  \cite{Ifrach2017}}
  \begin{center}
    \begin{tabular}{|c||c|}
    \hline
      Parameter & Value \\
      \hline
      $\beta$  & 0.925 \\
      $m$ &  0.125 \\
      $e$ & 75 \\
      $f$ & 10 \\
      $d$ & 4 \\
      $z$ & \{1,2\} \\
      $\delta$ & 0.7 \\
      $a$ & 1 \\
      \hline
    \end{tabular}
  \end{center}
\end{table}

\begin{figure}[h!]
\centering
\includegraphics[width=.6\textwidth]{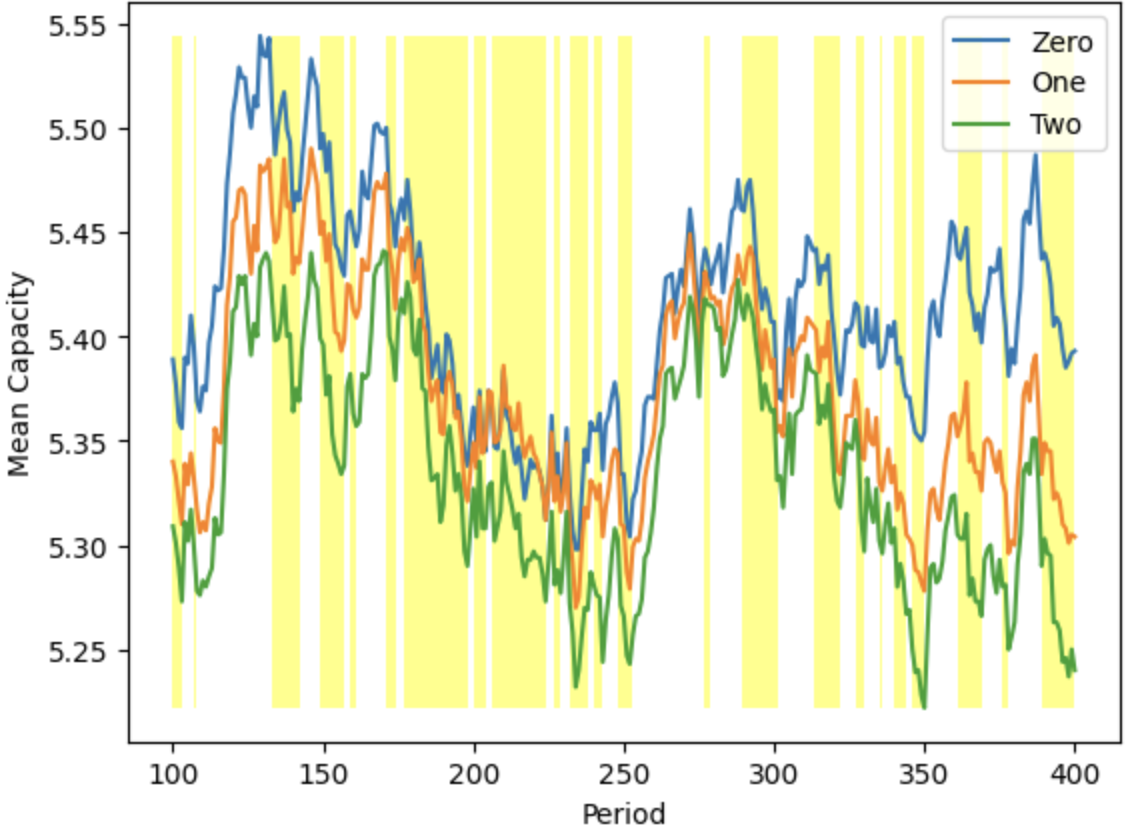}
\caption{We plot the mean aggregate capacity across the zero, one and two period foresight capacity  models. We simulate a single path of aggregate shocks for each of the models. The variability in mean capacity increases with the degree of foresight. In the two foresight case (green line), the variation to negative demand shocks appears greater than the zero foresight model (blue line). The models reach similar peaks in mean capacity during high demand shock states, but negative demand shocks results in greater decreases in mean capacity for two periods of foresight. The yellow part indicates low aggregate demand shock states, $z=1$, as opposed to high aggregate demand shock states, $z=2$. }
\label{fig:cap_sim}
\end{figure}

We compute the bounded foresight equilibrium in the case of zero, one and two periods of foresight. Similarly to the Krusell-Smith simulation, we observe a greater amount of variability in the endogenous economic variables, in this case mean firm capacity (which corresponds to the mean production in the economy under the model parameters we consider). Specifically, we observe that mean capacity is more responsive to high and low aggregate demand shocks with greater foresight. Figure \ref{fig:cap_sim} demonstrates the variability under the simulation of a single path of 400 periods. Table \ref{tab:cap_var} shows that the standard deviation of mean capacity over a simulated path of aggregate shocks is increasing in the degree of foresight. 

While this model differs significantly from the Krusell-Smith economy, the mechanism behind the increased variation in key economic variables is similar. Firms adjust their investment decisions more variably as they extend their foresight horizon, creating a feedback loop where greater awareness of  fluctuations leads to optimal policy adjustments that, in turn, amplify variability in market variables.

\begin{table}[h!]
\caption{The table lists the average standard deviation of mean capacity for each of the zero, one and two period foresight models. That is a for a single path of 1000 periods, we compute the standard deviation of the mean of the population capacity. We then do this over 1000 paths to compute the expected variability. We observe that the standard deviation in mean capacity increases with the degree of foresight.}
\label{tab:cap_var}
  \begin{center}
    \begin{tabular}{c||c|c|c} 
       & \textbf{Zero} & \textbf{One} & \textbf{Two} \\
      \hline
      $\sigma_{\bar{x}}$  & 0.087 & 0.093 & 0.101 \\
    \end{tabular}
  \end{center}
\end{table}

\section{Foresight Error}

In our setting, forecast errors are endogenous and depend on the model studied. While it may seem counterintuitive at first, our results show that forecast errors increase with additional foresight in equilibrium. This aligns with our numerical results discussed in the previous section, where key economic variables exhibit greater variation as foresight increases, leading to larger forecast errors. On the other hand, as the foresight horizon 
$N$ approaches infinity and the bounded foresight equilibrium becomes similar to rational expectations, forecast errors have to vanish.

The following three figures compare the computation of forecast errors. In figure \ref{fig:error}a, the blue tuples of population state and aggregate state string represents the forecast equilibrium variables along the tree for one and two periods ahead from the aggregate state $z_1$. Figure \ref{fig:error}b shows the realization of the aggregate states $z_1z_1$ in the black tuple of population state and aggregate state string and then figure \ref{fig:error}c shows the realization of aggregate state $z_1z_1z_2$, also in a black tuple. The forecast error for each model compares the forecasts from figure \ref{fig:error}a, that is the forecasts of the economic variables for aggregate states $z_1z_1$ and $z_1z_1z_2$ at time $z_1$, with the actual realization of those economic variables along that path. In the case of one period of foresight, the forecasts in aggregate states $z_1z_1$ and $z_1z_1z_2$  at $z_1$ are the same, but their realizations will differ. For two period foresight, both forecasts in those aggregate states will differ and for zero period foresight all forecasts will be the same as the one in $z_1$ and the realizations will differ. We then compute approximately the expected forecast error for each of the Krusell-Smith and Capacity Constraint models, for each of the zero, one and two foresight cases given in Equation \ref{Equation:forecast}.  In this equation, $\bar{s}^f_{z_{t+j}}$ refers to the forecasted mean of the population state $\bar{s}$ for the aggregate history $z_{t+j}$ from the perspective of the initial aggregate state $z_t$, while $\bar{s}^r_{z_{t+j}}$ refers to the realized mean population state for the same aggregate history $z_{t+j}$. That is, the forecast $\bar{s}^f_{z_{t+j}}$ captures the expected population state along a given future path as predicted at $z_t$, whereas $\bar{s}^r_{z_{t+j}}$ represents the actual realized population state along that path.

\begin{equation} \label{Equation:forecast}
\mathbb{E} \left[ \frac{\big|\bar{s}^f_{z_{t+j}}-\bar{s}^r_{z_{t+j}}\big|}{\bar{s}^r_{z_{t+j}}} \Bigg| z_t \right]
\end{equation}

\begin{figure}[H]
\centering
\begin{subfigure}{0.49\textwidth}
\resizebox{\linewidth}{!}{
\begin{tikzpicture}
     \draw (-5,0) -- (-3,1) ;
     \draw (-5,0) -- (-3,-1) ;
     \draw (-1,2) -- (1,3) ;
     \draw (-1,-2) -- (1,-3) ;
     \draw (-1,2) -- (1,1) ;
     \draw (-1,-2) -- (1,-1) ;
     \node[circle,draw=red, fill=red, inner sep=0pt,minimum size=12pt] at (-5,0) {};
     \node[circle,draw=black, fill=black, inner sep=0pt,minimum size=7pt] at (-5,0) {};
     \node at (-6.5,0) {$(\tilde{s}_{z_1};z_1)$};
     \node[text=blue] at (-2.5,1.5) {$(\tilde{s}_{z_1z_1};z_1z_1)$};
     \node[text=blue] at (-2.5,-1.5) {$(\tilde{s}_{z_1z_2};z_1z_2)$};
     \node[text=blue] at (1.5,3.25) {$(\tilde{s}_{z_1z_1c};z_1z_1z_1)$};
     \node[text=blue] at (1.5,-.75) {$(\tilde{s}_{z_1z_2c};z_1z_2z_1)$};
     \node[text=blue] at (1.5,.75) {$(\tilde{s}_{z_1z_1c};z_1z_1z_2)$};
     \node[text=blue] at (1.5,-3.25) {$(\tilde{s}_{z_1z_2c};z_1z_2z_2)$};
\end{tikzpicture}}
\caption{Current forecast}
\end{subfigure}
\begin{subfigure}{0.49\textwidth}
\resizebox{\linewidth}{!}{
\begin{tikzpicture}
     \draw[line width=2.5,red] (-5,0) -- (-3,1) ;
     \draw[line width=.5,black] (-5,0) -- (-3,1) ;
     \draw (-5,0) -- (-3,-1) ;
     \draw (-1,2) -- (1,3) ;
     \draw (-1,-2) -- (1,-3) ;
     \draw (-1,2) -- (1,1) ;
     \draw (-1,-2) -- (1,-1) ;
     \node[circle,draw=red, fill=red, inner sep=0pt,minimum size=12pt] at (-1,2) {};
     \node[circle,draw=black, fill=black, inner sep=0pt,minimum size=7pt] at (-1,2) {};
     \node at (-6.5,0) {$(\tilde{s}_{z_1};z_1)$};
     \node at (-2.5,1.5) {$(\tilde{s}_{z_1z_1};z_1z_1)$};
     \node[text=red] at (-2,-1.5) {$z_1z_2$};
     \node[text=blue] at (1.5,3.25) {$(\tilde{s}_{z_1z_1z_1};z_1z_1z_1)$};
     \node[text=red] at (1.5,-1.25) {$z_1z_2z_1$};
     \node[text=blue] at (1.5,.75) {$(\tilde{s}_{z_1z_1z_2};z_1z_1z_2)$};
     \node[text=red] at (1.5,-3.25) {$z_1z_2z_2$};
\end{tikzpicture}}    
\caption{Moving one period ahead}
\end{subfigure}
\begin{subfigure}{0.5\textwidth}
\resizebox{\linewidth}{!}{
\begin{tikzpicture}
     \draw[line width=2.5,red] (-5,0) -- (-3,1) ;
     \draw[line width=.5,black] (-5,0) -- (-3,1) ;
     \draw[line width=2.5,red] (-1,2) -- (1,1) ;
     \draw[line width=.5,black] (-1,2) -- (1,1) ;
     \draw (-5,0) -- (-3,-1) ;
     \draw (-1,2) -- (1,3) ;
     \draw (-1,-2) -- (1,-3) ;
     \draw (-1,2) -- (1,1) ;
     \draw (-1,-2) -- (1,-1) ;
     \node[circle,draw=red, fill=red, inner sep=0pt,minimum size=12pt] at (1.5,.75) {};
     \node[circle,draw=black, fill=black, inner sep=0pt,minimum size=7pt] at (1.5,.75) {};
     \node at (-6.5,0) {$(\tilde{s}_{z_1};z_1)$};
     \node at (-2.5,1.5) {$(\tilde{s}_{z_1z_1};z_1z_1)$};
     \node[text=red] at (-2,-1.5) {$z_1z_2$};
     \node[text=red] at (1.5,3.25) {$z_1z_1z_1$};
     \node[text=red] at (1.5,-1.25) {$z_1z_2z_1$};
     \node at (1.5,.25) {$(\tilde{s}_{z_1z_1z_2};z_1z_1z_2)$};
     \node[text=red] at (1.5,-3.25) {$z_1z_2z_2$};
\end{tikzpicture}}
\caption{Moving two periods ahead}
\end{subfigure}
\caption{The figure demonstrates the realization of a single simulated path between subfigures (a), (b) and (c) for a model with one-period of foresight. At aggregate state $z_1$ in (a), there is a forecast along each possible path. The forecast error compares the realization of $\tilde{s}$ along the realized simulated path in (b) and (c) in black, in comparison with the forecast made in (a) in blue. In red are the unrealized paths of aggregate states.}
\label{fig:error}
\end{figure}
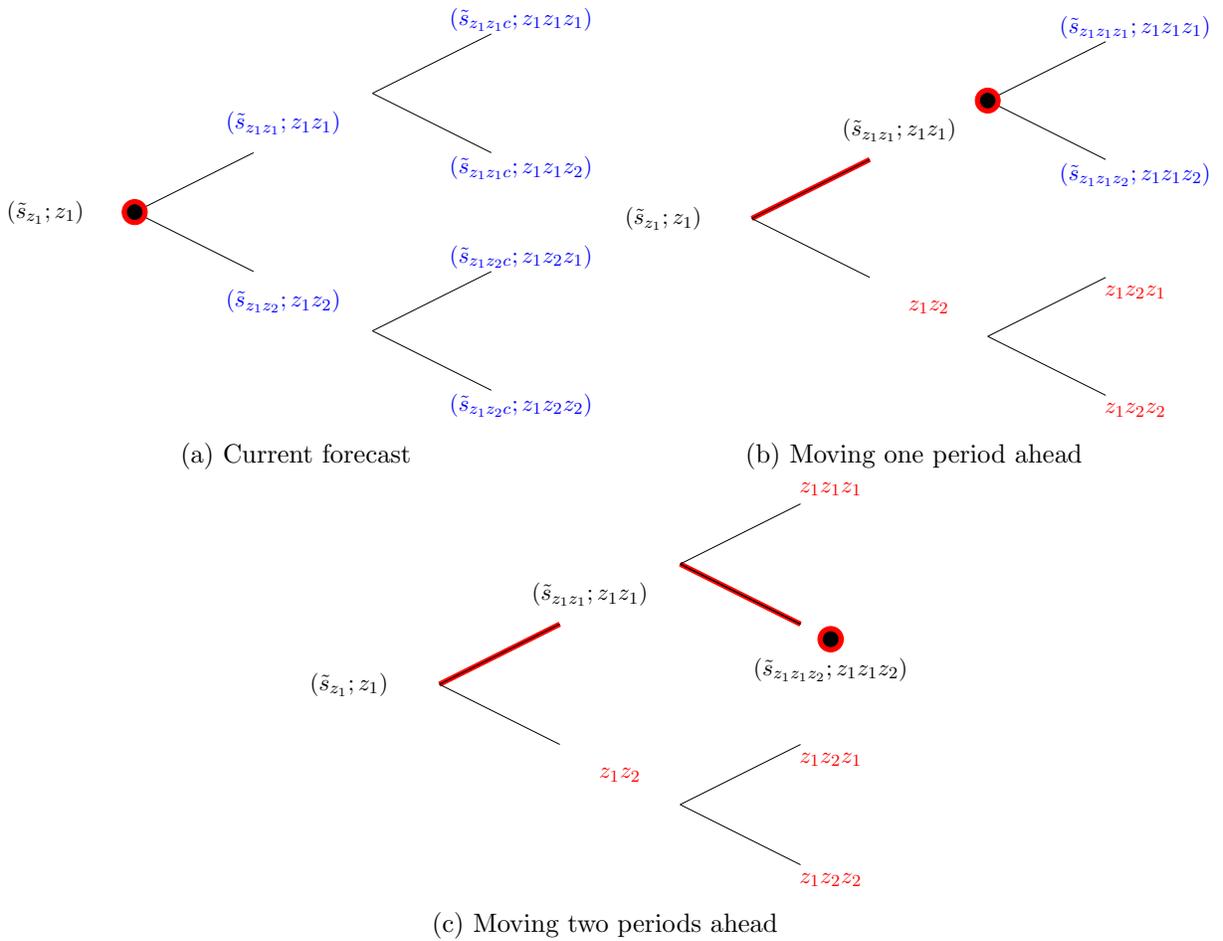

we find that forecast errors one, two, and three periods ahead increase as the foresight horizon expands in both the Krusell-Smith and Capacity Competition models. In other words, adding more foresight increases forecast errors for the same forecast horizon. 

In table \ref{tab:ks_error} and \ref{tab:cap_error}, we find that the forecast errors one, two and three periods ahead of the realized aggregate state increase along the foresight horizon in both the Krusell-Smith and Capacity Competition models. That is, an additional degree of foresight increases the forecast error for the same forecast horizon. As explained above, 
this occurs because greater foresight leads agents to internalize more variability in their decisions, which feeds back into both realized and forecasted equilibrium variables.  Thus, agent awareness of variability results in feedback into an increased realized variability but also in a greater difficulty of accurately forecasting that variability.

\begin{table}[h!]
  \begin{center}
    \begin{tabular}{c||c|c} 
      & \multicolumn{2}{c}{\textbf{Foresight Model}} \\
      Forecast Error & \textbf{Zero} & \textbf{One} \\
      \hline\\ [-2.5ex]
      \text{One Period} & 0.398\% & 1.080\%   \\
      \text{Two Period } & 0.671\% & 1.324\%  \\
      \text{Three Period} & 0.915\% & 1.515\%   \\
    \end{tabular}
  \end{center}
  \caption{We compare the foresight error for each of the zero and one foresight Krusell-Smith models. The degree of foresight shows an increasing amount of foresight error as well as along the horizon of the forecast, one to three periods ahead. We compute the average error along a single path and then take the average over 1000 paths to compute the expected foresight error.}
  \label{tab:ks_error}
\end{table}

\begin{table}[h!]
  \begin{center}
    \begin{tabular}{c||c|c|c} 
      & \multicolumn{3}{c}{\textbf{Foresight Model}} \\
      Forecast Error & \textbf{Zero} & \textbf{One} & \textbf{Two}\\
      \hline\\ [-2.5ex]
      \text{One Period} & 0.236\% & 0.308\% & 0.423\%  \\
      \text{Two Period } & 0.355\% & 0.420\% & 0.509\% \\
      \text{Three Period} & 0.453\% & 0.489\% & 0.650\%  \\
    \end{tabular}
  \end{center}
  \caption{We compare the foresight error for each of the zero, one and two foresight capacity constraint models. The degree of foresight shows an increasing amount of foresight error as well as along the horizon of the forecast, one to three periods ahead. We compute the average error along a single path and then take another average over 1000 paths to compute the expected foresight error.}
  \label{tab:cap_error}
\end{table}

\section{Conclusion}
This paper introduces N-Bounded Foresight Equilibrium (N-BFE) as a novel solution concept for large dynamic economies with heterogeneous agents and aggregate shocks. By relaxing the rational expectations assumption and instead allowing agents to form expectations over a limited foresight horizon, our framework significantly reduces the computational complexity of equilibrium analysis while maintaining forward-looking behavior. This approach aligns with both behavioral and computational perspectives, as it naturally incorporates behavioral features such as time inconsistency and limited attention, along with computational features such as constraining the planning horizon to a finite depth, similar to N-step lookahead methods in reinforcement learning. 

Our numerical analysis demonstrates that foresight is a key driver of equilibrium variation, highlighting its distinct role beyond traditional risk aversion or precautionary savings mechanisms. These numerical results suggest that foresight has important implications for economic variation, but further research is needed to explore its broader effects in different economic settings.

Finally, the flexibility of N-BFE makes it applicable to a wide range of models, including those used to study fiscal and monetary policy, such as HANK models. Future research can leverage this equilibrium framework to analyze policy  and aggregate fluctuations in these environments.

\clearpage
\bibliographystyle{ecta}
\bibliography{foresight}

\begin{thebibliography}{32}
\newcommand{\enquote}[1]{``#1''}
\expandafter\ifx\csname natexlab\endcsname\relax\def\natexlab#1{#1}\fi

\bibitem[\protect\citeauthoryear{Acemoglu and Jensen}{Acemoglu and
  Jensen}{2015}]{acemoglu2015robust}
\textsc{Acemoglu, D. and M.~K. Jensen} (2015): \enquote{Robust comparative
  statics in large dynamic economies,} \emph{Journal of Political Economy},
  123, 587--640.

\bibitem[\protect\citeauthoryear{Adlakha, Johari, and Weintraub}{Adlakha
  et~al.}{2015}]{adlakha2015equilibria}
\textsc{Adlakha, S., R.~Johari, and G.~Y. Weintraub} (2015):
  \enquote{Equilibria of Dynamic Games with Many Players: Existence,
  Approximation, and Market Structure,} \emph{Journal of Economic Theory}.

\bibitem[\protect\citeauthoryear{Ahn, Kaplan, Moll, Winberry, and Wolf}{Ahn
  et~al.}{2018}]{ahn2018inequality}
\textsc{Ahn, S., G.~Kaplan, B.~Moll, T.~Winberry, and C.~Wolf} (2018):
  \enquote{When inequality matters for macro and macro matters for inequality,}
  \emph{NBER macroeconomics annual}, 32, 1--75.

\bibitem[\protect\citeauthoryear{Aiyagari}{Aiyagari}{1994}]{aiyagari1994uninsured}
\textsc{Aiyagari, S.~R.} (1994): \enquote{Uninsured idiosyncratic risk and
  aggregate saving,} \emph{The Quarterly Journal of Economics}, 109, 659--684.

\bibitem[\protect\citeauthoryear{Aliprantis and Border}{Aliprantis and
  Border}{2006}]{aliprantis2006infinite}
\textsc{Aliprantis, C.~D. and K.~Border} (2006): \emph{Infinite Dimensional
  Analysis: a hitchhiker's guide}, Springer.

\bibitem[\protect\citeauthoryear{Auclert, Bard{\'o}czy, Rognlie, and
  Straub}{Auclert et~al.}{2021}]{auclert2021using}
\textsc{Auclert, A., B.~Bard{\'o}czy, M.~Rognlie, and L.~Straub} (2021):
  \enquote{Using the sequence-space Jacobian to solve and estimate
  heterogeneous-agent models,} \emph{Econometrica}, 89, 2375--2408.

\bibitem[\protect\citeauthoryear{Bertsekas}{Bertsekas}{2019}]{bertsekas2019reinforcement}
\textsc{Bertsekas, D.} (2019): \emph{Reinforcement learning and optimal
  control}, vol.~1, Athena Scientific.

\bibitem[\protect\citeauthoryear{Bertsekas and Shreve}{Bertsekas and
  Shreve}{1996}]{bertsekas1996stochastic}
\textsc{Bertsekas, D. and S.~E. Shreve} (1996): \emph{Stochastic optimal
  control: the discrete-time case}, vol.~5, Athena Scientific.

\bibitem[\protect\citeauthoryear{Besanko and Doraszelski}{Besanko and
  Doraszelski}{2004}]{besanko2004capacity}
\textsc{Besanko, D. and U.~Doraszelski} (2004): \enquote{Capacity Dynamics and
  Endogenous Asymmetries in Firm Size,} \emph{RAND Journal of Economics},
  23--49.

\bibitem[\protect\citeauthoryear{Besanko, Doraszelski, Lu, and
  Satterthwaite}{Besanko et~al.}{2010}]{besanko2010lumpy}
\textsc{Besanko, D., U.~Doraszelski, L.~X. Lu, and M.~Satterthwaite} (2010):
  \enquote{Lumpy Capacity Investment and Disinvestment Dynamics,}
  \emph{Operations Research}, 58, 1178--1193.

\bibitem[\protect\citeauthoryear{Bilal}{Bilal}{2023}]{bilal2023solving}
\textsc{Bilal, A.} (2023): \enquote{Solving heterogeneous agent models with the
  master equation,} Tech. rep., National Bureau of Economic Research.

\bibitem[\protect\citeauthoryear{Boppart, Krusell, and Mitman}{Boppart
  et~al.}{2018}]{boppart2018exploiting}
\textsc{Boppart, T., P.~Krusell, and K.~Mitman} (2018): \enquote{Exploiting MIT
  shocks in heterogeneous-agent economies: the impulse response as a numerical
  derivative,} \emph{Journal of Economic Dynamics and Control}, 89, 68--92.

\bibitem[\protect\citeauthoryear{Broer, Kohlhas, Mitman, and Schlafmann}{Broer
  et~al.}{2022}]{broer2022possibility}
\textsc{Broer, T., A.~N. Kohlhas, K.~Mitman, and K.~Schlafmann} (2022):
  \enquote{On the possibility of krusell-smith equilibria,} \emph{Journal of
  Economic Dynamics and Control}, 141, 104391.

\bibitem[\protect\citeauthoryear{Caplin and Nalebuff}{Caplin and
  Nalebuff}{1991}]{caplin1991aggregation}
\textsc{Caplin, A. and B.~Nalebuff} (1991): \enquote{Aggregation and Imperfect
  Competition: On the Existence of Equilibrium,} \emph{Econometrica}, 25--59.

\bibitem[\protect\citeauthoryear{Ericson and Pakes}{Ericson and
  Pakes}{1995}]{ericson1995markov}
\textsc{Ericson, R. and A.~Pakes} (1995): \enquote{Markov-Perfect Industry
  Dynamics: A Framework for Empirical Work,} \emph{The Review of Economic
  Studies}, 53--82.

\bibitem[\protect\citeauthoryear{Fern{\'a}ndez-Villaverde, Hurtado, and
  Nuno}{Fern{\'a}ndez-Villaverde et~al.}{2023}]{fernandez2023financial}
\textsc{Fern{\'a}ndez-Villaverde, J., S.~Hurtado, and G.~Nuno} (2023):
  \enquote{Financial frictions and the wealth distribution,}
  \emph{Econometrica}, 91, 869--901.

\bibitem[\protect\citeauthoryear{Gabaix}{Gabaix}{2024}]{Gabaix2024}
\textsc{Gabaix, X.} (2024): \enquote{Behavioral Macroeconomics Via Sparse
  Dynamic Programming,} \emph{Journal of the European Economic Association}.

\bibitem[\protect\citeauthoryear{Halevy}{Halevy}{2015}]{halevy2015time}
\textsc{Halevy, Y.} (2015): \enquote{Time consistency: Stationarity and time
  invariance,} \emph{Econometrica}, 83, 335--352.

\bibitem[\protect\citeauthoryear{Huggett}{Huggett}{1993}]{huggett1993risk}
\textsc{Huggett, M.} (1993): \enquote{The risk-free rate in heterogeneous-agent
  incomplete-insurance economies,} \emph{Journal of economic Dynamics and
  Control}, 17, 953--969.

\bibitem[\protect\citeauthoryear{Ifrach and Weintraub}{Ifrach and
  Weintraub}{2017}]{Ifrach2017}
\textsc{Ifrach, B. and G.~Y. Weintraub} (2017): \enquote{A Framework for
  Dynamic Oligopoly in Concentrated Industries,} \emph{The Review of Economic
  Studies}, 84, 1106--1150.

\bibitem[\protect\citeauthoryear{Kaplan, Moll, and Violante}{Kaplan
  et~al.}{2018}]{kaplan2018monetary}
\textsc{Kaplan, G., B.~Moll, and G.~L. Violante} (2018): \enquote{Monetary
  policy according to HANK,} \emph{American Economic Review}, 108, 697--743.

\bibitem[\protect\citeauthoryear{Krusell and Smith}{Krusell and
  Smith}{1998}]{krusell1998income}
\textsc{Krusell, P. and A.~A. Smith, Jr} (1998): \enquote{Income and wealth
  heterogeneity in the macroeconomy,} \emph{Journal of political Economy}, 106,
  867--896.

\bibitem[\protect\citeauthoryear{Laibson}{Laibson}{1997}]{laibson1997}
\textsc{Laibson, D.} (1997): \enquote{Golden Eggs and Hyperbolic Discounting,}
  \emph{The Quarterly Journal of Economics}, 112, 443--477.

\bibitem[\protect\citeauthoryear{Light}{Light}{2020}]{Light2017}
\textsc{Light, B.} (2020): \enquote{Uniqueness of Equilibrium in a
  Bewley-Aiyagari Model,} \emph{Economic Theory}, 69, 435--450.

\bibitem[\protect\citeauthoryear{Light}{Light}{2024}]{light2024principle}
---\hspace{-.1pt}---\hspace{-.1pt}--- (2024): \enquote{The Principle of
  Optimality in Dynamic Programming: A Pedagogical Note,} \emph{Operations
  Research Letters}, 107164.

\bibitem[\protect\citeauthoryear{Light and Weintraub}{Light and
  Weintraub}{2022}]{light2022mean}
\textsc{Light, B. and G.~Y. Weintraub} (2022): \enquote{Mean field equilibrium:
  uniqueness, existence, and comparative statics,} \emph{Operations Research},
  70, 585--605.

\bibitem[\protect\citeauthoryear{Moll}{Moll}{2024}]{Moll2024}
\textsc{Moll, B.} (2024): \enquote{The Trouble with Rational Expectations in
  Heterogeneous Agent Models: A Challenge for Macroeconomics,} \emph{Working
  Paper}.

\bibitem[\protect\citeauthoryear{Pakes and McGuire}{Pakes and
  McGuire}{1994}]{pakes1994computing}
\textsc{Pakes, A. and P.~McGuire} (1994): \enquote{Computing Markov-Perfect
  Nash Equilibria: Numerical Implications of a Dynamic Differentiated Product
  Model,} \emph{The Rand Journal of Economics}, 555--589.

\bibitem[\protect\citeauthoryear{Serfozo}{Serfozo}{1982}]{serfozo1982convergence}
\textsc{Serfozo, R.} (1982): \enquote{Convergence of Lebesgue Integrals with
  Varying Measures,} \emph{Sankhy{\=a}: The Indian Journal of Statistics,
  Series A}, 380--402.

\bibitem[\protect\citeauthoryear{Silver, Huang, Maddison, Guez, Sifre, Van
  Den~Driessche, Schrittwieser, Antonoglou, Panneershelvam, Lanctot
  et~al.}{Silver et~al.}{2016}]{silver2016mastering}
\textsc{Silver, D., A.~Huang, C.~J. Maddison, A.~Guez, L.~Sifre, G.~Van
  Den~Driessche, J.~Schrittwieser, I.~Antonoglou, V.~Panneershelvam,
  M.~Lanctot, et~al.} (2016): \enquote{Mastering the game of Go with deep
  neural networks and tree search,} \emph{nature}, 529, 484--489.

\bibitem[\protect\citeauthoryear{Weintraub, Benkard, and Van~Roy}{Weintraub
  et~al.}{2008}]{weintraub2008markov}
\textsc{Weintraub, G.~Y., C.~L. Benkard, and B.~Van~Roy} (2008):
  \enquote{Markov Perfect Industry Dynamics with Many Firms,}
  \emph{Econometrica}, 1375--1411.

\bibitem[\protect\citeauthoryear{Weintraub, Benkard, and Van~Roy}{Weintraub
  et~al.}{2010}]{weintraub2010computational}
---\hspace{-.1pt}---\hspace{-.1pt}--- (2010): \enquote{Computational Methods
  for Oblivious Equilibrium,} \emph{Operations research}, 1247--1265.

\end{thebibliography}

\break

\appendix

\section{Appendix}

\begin{proof}[Proof of Proposition \ref{Prop:ValueF}]
    First, suppose that $\tau \geq N$. Then $s_{\tau} = s^{C}$ for some constant population state and every population state after the period $\tau$ is also assumed to equal $s^{C}$. Then given Assumption \ref{Assumption:1}, standard dynamic programming arguments (see the continuous dynamic programming analysis in \cite{light2024principle}) show that the value function is the unique and continuous function that satisfies the Bellman equation 
    $$ V^{\mathcal{I}(t,N,z_{t})}\left(x_{\tau}, z_{\tau}, s^{C} \right ) =\max _{a \in \Gamma (x_{\tau})}\pi (x_{\tau } ,z_{\tau} ,a ,s^{C}) + \beta\int _{E} \sum _{z_{i}}P(z_{\tau
 },z_{i}) V^{\mathcal{I}(t,N,z_{t})} \left (w(x_{\tau} ,z_{i} ,a ,s^{C} , \zeta) , z_{i} , s^{C}  \right )q(d
\zeta). $$
Now continuing recursively from period $\tau$ backwards to period $t$ shows that $V^{\mathcal{I}(t,N,z_{t})}$ satisfies the equation in Proposition \ref{Prop:ValueF}. From Assumption 1 and the maximum theorem (see Theorem 17.31 in \cite{aliprantis2006infinite}), $g^{\mathcal{I}(t,N,z_{t})}\left(x_{\tau}, z_{\tau}, s_{\tau}(z_{t} ,\ldots  ,z_{\tau }) \right )$ is upper hemi-continuous so it is a continuous function when it is single-valued. 
\end{proof}

In order to establish the existence of an equilibrium, we will use the following two Propositions.

We say that $k_{n} :X \rightarrow \mathbb{R}$ converges continuously to $k$ if $k_{n} (x_{n}) \rightarrow k (x)$ whenever $x_{n} \rightarrow x$. The following Proposition is a special case of Theorem 3.3 in \cite{serfozo1982convergence}.

\begin{proposition}
\label{Prop1}Assume that $k_{n} :X \rightarrow \mathbb{R}$ is a uniformly bounded sequence of functions. If $k_{n} :X \rightarrow \mathbb{R}$ converges continuously to $k$ and $s_{n}$ converges weakly to $s$ then
\begin{equation*}\underset{n \rightarrow \infty }{\lim }\int _{X}k_{n} (x) s_{n} (d x) =\int _{X}k (x) s (d x).
\end{equation*}
\end{proposition}

For a proof of the following proposition see Corollary 17.56 in \cite{aliprantis2006infinite}.

\begin{proposition}
(Schauder-Tychonoff) Let $K$ be a nonempty, compact, convex subset of a locally convex Hausdorff space, and let $f :K \rightarrow K$ be a continuous function. Then the set of fixed points of $f$ is compact and nonempty. 
\end{proposition}

\begin{proof}
[Proof of Theorem~\ref{Theorem existence}] Let $g  \in G $ be a continuous optimal policy. 

If $N=0$ the result is trivial so assume $N \geq 1$. 
Let $m= 1 + \sum _{j=1}^{N} |Z|^{j}$ where $|Z|$ is the number of states in the set $Z$. Let $z_{t} \in Z$ be the state in period $t$ and consider the operator $\Phi  :\mathcal{P} (X)^{m} \rightarrow \mathcal{P} (X)^{m}$ defined for  $\boldsymbol{s}  \in \mathcal{P} (X)^{m}$ by $\Phi s (z_{t}) = s(z_{t})$ and

 \begin{align} 
\Phi s(z^{t ,t +k} ,z_{i}) (D) &=\int _{X}
 \int _{E} 
 & 1_{D} \left (w(x,z_{i }, g^{\mathcal{I}(t,N,z_{t})}(x ,z_{t+k}, ,s(z^{t ,t +k} ) ), s(z^{t ,t +k} ), \zeta ) \right ) q(d\zeta)s(z^{t ,t +k} )(dx) 
\end{align}
for any $D \in \mathcal{B}(X)$ and any history $z^{t ,t +k}$ with $k =0, \ldots ,N-1$ that starts with $z_{t}$. Here, $z_{i} \in Z$, and $\mathcal{I}(t,N,z_{t}) =(s_{t}(z_{t}) ,s_{t +1}(z^{t ,t +1}) ,\ldots  ,s_{t +N}(z^{t ,t +N}) ,\ldots ) = (\boldsymbol{s},\ldots)$. That is, $\Phi \boldsymbol{s}$ describes the dynamics of the population states during the first $N$ periods, assuming the agents' policy function is based on the population states given by $\boldsymbol{s}$.

If $\boldsymbol{s}$ is a fixed point of $\Phi $ then $\boldsymbol{s}$ is an N-BFE. Since $X$ is compact,  $\mathcal{P} (X)$ is compact (i.e., compact in the weak topology, see \cite{aliprantis2006infinite}). Thus, $\mathcal{P} (X)^{m}$ is compact. 
Clearly $\mathcal{P} (X)^{m}$ is convex and, endowed with the weak topology, is a locally convex Hausdorff space \citep{aliprantis2006infinite}. 

Thus. if $\Phi $ is continuous, we can apply Schauder-Tychonoff's fixed point theorem to prove that $\Phi $ has a fixed point. We now show that $\Phi $ is continuous.

Assume that $\boldsymbol{s}_{n}$ converges weakly to $\boldsymbol{s}$. Let $f :X \rightarrow \mathbb{R}$ be a continuous and bounded function. Establishing weak convergence coordinate-wise is sufficient to establish weak convergence in the product space $\mathcal{P}(X)^{m}$ so we let $z^{t,t+k}$ to be some history starting from $z_{t}$.

To simplify notation we write $g(x,z_{t+k},\boldsymbol{s})$ instead of  $g^{\mathcal{I}(t,N,z_{t})}(x ,z_{t+k}, ,s(z^{t ,t +k} ))$ 
to capture the dependence the policy function in the predicted population states. 

Since $w$ is jointly continuous and $g$ is continuous, we have 
$$f (w (x_{n} , z_{i}, g (x_{n}, z_{t+k} ,\boldsymbol{s}_{n}) ,{s}_{n}(z^{t,t+k}) ,\zeta ) \rightarrow f(w (x ,z_{i}, g (x ,z_{t+k},\boldsymbol{s}) ,s(z^{t,t+k}) ,\zeta )$$ whenever $x_{n} \rightarrow x$. 
Let $$k_{n} (x):=\int _{E}f (w (x_{n} , z_{i}, g (x_{n}, z_{t+k} ,\boldsymbol{s}_{n}) ,{s}_{n}(z^{t,t+k}) ,\zeta ) q(d\zeta)$$ 
and 
$$k (x):=\int _{E} f(w (x ,z_{i}, g (x ,z_{t+k},\boldsymbol{s}) ,{s}(z^{t,t+k}) ,\zeta )q(d\zeta).$$ 
Then $k_{n}$ converges continuously to $k$, i.e., $k_{n} (x_{n}) \rightarrow k (x)$ whenever $x_{n} \rightarrow x$. Since $f$ is bounded, the sequence $k_{n}$ is uniformly bounded. 

Thus, from Proposition \ref{Prop1},  we have 
\begin{align*}\underset{n \rightarrow \infty }{\lim }\int _{X}f (x) \Phi s_{n}(z^{t ,t +k} ,z_{i}) (dx) &  =\underset{n \rightarrow \infty }{\lim }\int _{X}k_{n} (x)s_{n}(z^{t ,t +k})(d x) \\
 &  =\int _{X}k (x) s(z^{t ,t +k})(d x)  \\
 &  =\int _{X}f (x) \Phi  s (z^{t ,t +k} ,z_{i}) (dx).
 \end{align*}
 We conclude that $\Phi  \boldsymbol{s}_{n}$ converges weakly to $\Phi \boldsymbol{s}$. Thus, $\Phi $ is continuous. From the Schauder-Tychonoff's fixed point theorem, $\Phi $ has a fixed point.

\end{proof}

\end{document}